\documentclass[10pt]{amsart}
\usepackage{amsmath,amssymb}

\begin{document}

\newtheorem{thm}{Theorem}[section]
\newtheorem{lem}[thm]{Lemma}
\newtheorem{prop}[thm]{Proposition}
\newtheorem{cor}[thm]{Corollary}
\newtheorem{defn}[thm]{Definition}
\newtheorem*{remark}{Remark}
\newtheorem{notn}[thm]{Notation}

\numberwithin{equation}{section}

\newcommand{\Z}{{\mathbb Z}} 
\newcommand{\Q}{{\mathbb Q}}
\newcommand{\R}{{\mathbb R}}
\newcommand{\C}{{\mathbb C}}
\newcommand{\G}{{\mathbb G}}
\newcommand{\N}{{\mathbb N}}
\newcommand{\FF}{{\mathbb F}}
\newcommand{\T}{{\mathbb T}}
\newcommand{\A}{{\mathbb A}}
\newcommand{\fq}{\mathbb{F}_q}
\newcommand{\E}{{\mathbb E}}
\newcommand{\PP}{{\mathbb P}}

\def\scrA{{\mathcal A}}
\def\scrB{{\mathcal B}}
\def\scrC{{\mathcal C}}
\def\scrD{{\mathcal D}}
\def\scrE{{\mathcal E}}
\def\scrH{{\mathcal H}}
\def\scrK{{\mathcal K}}
\def\scrL{{\mathcal L}}
\def\scrM{{\mathcal M}}
\def\scrN{{\mathcal N}}
\def\scrO{{\mathcal O}}
\def\scrPP{{\mathcal P}} 
\def\scrS{{\mathcal S}}
\def\x{{\underline{x}}}
\def\X{{\underline{X}}}

\newcommand{\rmk}[1]{\footnote{{\bf Comment:} #1}}

\renewcommand{\mod}{\;\operatorname{mod}}
\newcommand{\ord}{\operatorname{ord}}
\newcommand{\TT}{\mathbb{T}}
\renewcommand{\i}{{\mathrm{i}}}
\renewcommand{\d}{{\mathrm{d}}}
\renewcommand{\^}{\widehat}
\newcommand{\HH}{\mathbb H}
\newcommand{\Vol}{\operatorname{vol}}
\newcommand{\area}{\operatorname{area}}
\newcommand{\tr}{\operatorname{tr}}
\newcommand{\norm}{\mathcal N} 
\newcommand{\intinf}{\int_{-\infty}^\infty}
\newcommand{\ave}[1]{\left\langle#1\right\rangle} 
\newcommand{\Var}{\operatorname{Var}}
\newcommand{\Cov}{\operatorname{Cov}}
\newcommand{\Prob}{\operatorname{Prob}}
\newcommand{\sym}{\operatorname{Sym}}
\newcommand{\disc}{\operatorname{disc}}
\newcommand{\CA}{{\mathcal C}_A}
\newcommand{\cond}{\operatorname{cond}} 
\newcommand{\lcm}{\operatorname{lcm}}
\newcommand{\Kl}{\operatorname{Kl}} 
\newcommand{\leg}[2]{\left( \frac{#1}{#2} \right)}  
\newcommand{\SL}{\operatorname{SL}}
\newcommand{\Id}{\operatorname{Id}}
\newcommand{\diag}{\operatorname{diag}}
\newcommand{\diam}{\operatorname{diam}}
\newcommand{\supp}{\operatorname{supp}}

\newcommand{\sumstar}{\sideset \and^{*} \to \sum}

\newcommand{\LL}{\mathcal L} 
\newcommand{\sumf}{\sum^\flat}
\newcommand{\Hgev}{\mathcal H_{2g+2,q}}
\newcommand{\USp}{\operatorname{USp}}
\newcommand{\conv}{*}
\newcommand{\dist} {\operatorname{dist}}
\newcommand{\CF}{c_0} 
\newcommand{\kerp}{\mathcal K}

\newcommand{\gp}{\operatorname{gp}}
\newcommand{\Area}{\operatorname{Area}}

\title[Eigenfunctions on tori with random impurities]{Uniformly distributed eigenfunctions \\ on flat tori with random impurities}
\author{Henrik Uebersch\"ar}
\address{Laboratoire Paul Painlev\'e, Universit\'e Lille 1, 59655 Villeneuve d'Ascq, France.}
\email{henrik.ueberschar@math.univ-lille1.fr}
\date{\today}
\maketitle

\begin{abstract}
We study a random Schr\"odinger operator, the Laplacian with $N$ independently uniformly distributed random delta potentials on flat tori $\T^d_L=\R^d/L\Z^d$, $d=2,3$, where $L>0$ is large. We determine a condition in terms of the size of the torus $L$, the density of the potentials $\rho=NL^{-d}$ and the energy of the eigenfunction $E$ such that any such eigenfunctions will with nonzero probability be equidistributed on the entire torus. We remark that the equidistribution we prove here is still consistent with a localized regime, where the localization length is much larger than the size of the torus. In fact our result implies a certain polynomial lower bound on the localization length, so that the localization length becomes infinitely large as $E\to\infty$ or $\rho\to0$.
\end{abstract}

\section{Introduction}

As was first observed by Anderson \cite{A} the long-term dynamics of a wave packet in a random lattice of impurities can be spatially confined in the presence of sufficiently strong disorder. This phenomenon, known as ``Anderson localization'', is generally expected to occur when the wavelength is of size comparable to the elastic mean free path length. 

The physical interpretation is that in this localized regime the quantum particle ``feels'' the effect of scattering from the impurities which leads to an exponential decay in the low energy eigenfunctions of the system at distances significantly larger than the mean free path length. On the other hand if the wavelength is much smaller than the mean free path length (e. g. consider high energy eigenfunctions, or a low density of impurities) then the question is whether there exists a ``delocalized regime''.

The scaling theory of Abrahams, Anderson, Licciardello and Ramakrishnan \cite{AALR} predicts that the localization properties of the eigenfunctions of a disordered quantum system, as described above, ought to depend on the dimension of the system. 

Whereas in dimension $d=1$ exponential localization is always expected, independently of the strength of disorder, one expects a phase transition from localization at strong disorder/low energy to delocalization at weak disorder/high energy in dimension $d=3$. The $2$-dimensional case is critical, although, generally, exponential localization is always expected to occur as in dimension $d=1$. 


The present paper studies flat tori $\T^d$, $d=2,3$, with independently uniformly distributed random impurities, modeled by Dirac delta potentials\footnote{A rigorous realization of the formal Hamiltonian via the theory of self-adjoint extensions is only possible in dimension $d\leq 3$}. This means given a torus $\T^d_L=\R^d/L\Z^d$, where $L>0$ is a large parameter, we sample $N$ points $x_1,\cdots,x_N$, where $N$ is large, independently from a uniform distribution on $\T^d_L$, i. e. the $x_j$ are i.i.d. uniform random variables on $\T^d_L$. 

Let us consider the formal random Schr\"odinger operator
\begin{equation}
H_{\x_L}=-\Delta+\sum_{j=1}^N\alpha_j\delta(x-x_j), \quad \x_L=(x_1,\cdots, x_N),
\quad \forall j: \alpha_j\in\R
\end{equation}
which models a disordered quantum system (say an electron in a box with $N$ randomly distributed nuclei). We denote the density of the impurities by $\rho_L=N/L^d$, where the number of impurities may depend on the size of the torus, $N=N(L)$. 

The eigenfunctions of the random operator $H_{\x_L}$ are expected to be exponentially localized in configuration space at the bottom of the spectrum (for localization results regarding delta potentials cf. for instance \cite{BG} and \cite{HKK} or for smooth Poisson potentials \cite{GHK}). The inverse of the exponent in the exponential bound is called ``localization length'' and we will denote it by $L_{loc}$.

If we hold the density of impurities fixed and increase the energy, a question of great interest is whether we see a transition from localization to delocalization in the spatial geometry of the eigenfunctions. This means there should be a critical value for the energy $E_c$ such that for $E>E_c$ there exist eigenfunctions which are extended across the entire torus. If we are in the localized regime, we should be able to observe the localization on a large torus $\T^d_L$ if we fix the energy $E$ and make $L$ sufficiently large (the localization length may be large and depend on $E$ and $\rho$). In the delocalized regime however (i.e. $E>E_c$), no exponential localization will be observed on any torus $\T^d_L$ for fixed $E$, no matter how large $L$ is.

By a scaling argument, this problem can easily be seen to be equivalent (see subsection \ref{scaling}) to the delocalization at high energy of the eigenfunctions on a fixed size torus $\T^d=\R^d/\Z^d$ with random impurities. The density of impurities is given by $\rho=N$, and the localization properties of the eigenfunctions depend on the density $\rho$ and the eigenvalue $\lambda$. 

Let $\x=(x_1,\cdots,x_N)$ be i.i.d. uniform random variables on $\T^d$. 
We consider the formal random Schr\"odinger operator
\begin{equation}\label{random Schrodinger}
H_\x=-\Delta+\sum_{j=1}^N\alpha_j\delta(x-x_j)
\end{equation}
which may be realized rigorously by applying the theory of self-adjoint extensions (see subsection \ref{extensions}) to the restricted Laplacian $-\Delta|_{C^\infty_c(\T^d-\x)}$. 

We denote the family of self-adjoint extensions associated with the formal operator \eqref{random Schrodinger} by $\{-\Delta_{\x,U}\}_{U\in U(N)}$. The number of self-adjoint extensions exceeds the number of physical coupling constants. We remark that in particular the subgroup of diagonal unitary matrices $D(N)\subset U(N)$ corresponds to the case where a non-local interaction between the individual impurities is forbidden.

For given $U\in U(N)$ the operator $-\Delta_{\x,U}$ has three types of eigenfunctions:
\begin{itemize}
\item[1.] ``Old eigenfunctions'' of the Laplacian which vanish at all the points $x_j$, $j=1,\cdots,N$, and therefore do not ``feel'' the effect of any of the impurities.\\
\item[2a.] ``Non-generic new eigenfunctions'' which vanish at some, but not all, of the points $x_j$. They arise in subspaces of eigenspaces of lower rank perturbations of the Laplacian. Their occurrence constitutes a probability zero event.\\
\item[2b.] ``Generic new eigenfunctions'' which do not vanish at any of the points $x_j$, rather diverge logarithmically near each of the locations of the impurities. These eigenfunctions feel the effect of all impurities and are the objects of study in this paper.
\end{itemize}

Since the operator $H_\x$ is a rank $N$ perturbation of the Laplacian, it has at most $N$ new eigenfunctions corresponding to new eigenvalues which are ``torn off'' each old eigenspace of the Laplacian, provided the dimension of the eigenspace is large enough. 

The eigenvalues of the Laplacian on the torus $\T^d=\R^d/\Z^d$ are given by the set $S=\{n \mid n=4\pi^2(x_1^2+\cdots+x_d^2), \;x_1,\cdots,x_d\in\Z\}=\{0=n_0<n_1<n_2<\cdots\}$. The multiplicity of a Laplacian eigenvalue $n$ is given by the number of ways the integer $n/4\pi^2$ can be written as a sum of $d$ squares. 

If $d=2$ the multiplicity of $n$ grows on average like $\sqrt{\log n}$, which is a consequence of Landau's Theorem \cite{L}: $$\#\{n\in S \mid n\leq x\}\sim \frac{Bx}{\sqrt{\log x}}.$$ If $d=3$, the multiplicity grows on average like $n^{1/2}$.

We have the following theorem which proves that with positive probability there exist uniformly distributed eigenfunctions of the random operator $-\Delta_{\x,U}$ for sufficiently high energy. Note that for a given new eigenvalue $\lambda$ we have almost surely $\lambda\in(n_k,n_{k+1})$ for some $n_k\in S$ (for a detailed explanation see subsection 2.1.1).
\begin{thm}\label{deloc}
Fix $U\in U(N)$. Let $a\in C^\infty(\T^d)$. Denote by $\{g^N_{\lambda,\x}\}$ the $L^2$-normalized generic new eigenfunctions of the random operator $-\Delta_{\x,U}$. There exists a subsequence $S'\subset S$ of density $1$ and constants $\gamma_d,\lambda_0>0$ such that the following holds: If the points $x_j$, $j=1,\cdots,N$ are i.i.d. uniform random variables on $\T^d$ and $n_k\in S'$, $n_k\geq\lambda_0$, we have for each new eigenvalue $\lambda\in(n_k,n_{k+1})$ with probability $\gtrsim\frac{1}{N}$\footnote{The notation $f\gtrsim g$ denotes $\exists C>0: f\geq Cg$.} that
\begin{equation}
\int_{\T^d}a(y)\Big|g^N_{\lambda,\x}(y)\Big|^2 dy = \int_{\T^d}a(y)dy+\scrO_\epsilon(\|\hat{a}\|_{l^1}N^{1/2}\lambda^{-\gamma_d+\epsilon}).
\end{equation}
\end{thm}
\begin{remark}
To be precise, $\gamma_2=\frac{17}{832}$. Furthermore, $\gamma_3=\frac{1}{12}$ for the cubic lattice $\Z^3$ 
(in fact this exponent equals the optimal exponent in section 3 of the paper \cite{Y}, where one has to choose $\delta=\frac{1}{6}$ to get this optimal exponent).
\end{remark}

By a straightforward scaling argument (cf. section \ref{scaling}) we obtain the main result of this paper.
\begin{thm}
Let $\T^d_L=\R^d/L\Z^d$. Let $\x_L$ be an $N$-point uniform random process on $\T^d_L$\footnote{I.e. the $x_j$, $j=1,\cdots,N$, are i.i.d. uniform random variables on $\T^d_L$.}. Fix $U\in U(N)$ and denote by $-\Delta_{\x_L,U}$ the corresponding self-adjoint extension of $-\Delta|_{C^\infty_c(\T^d_L-\x_L)}$. Denote an $L^2$-normalized generic new eigenfunction of the random operator $-\Delta_{\x_L,U}$ with eigenvalue $E$ by $\psi_E$. 

Let $a\in C^\infty(\T^d_L)$. 
If we sample the points $x_1,\cdots x_N\in\x_L$, and $n_k\in S'$, $n_k\geq\lambda_0$, we have for each $E$ such that $EL^2\in(n_k,n_{k+1})$ with probability $\gtrsim\frac{1}{N}$,
\begin{equation}\label{unifdist}
\int_{\T^d_L}a(y)|\psi_E(y)|^2 dy = \frac{1}{L^d}\int_{\T^d_L}a(y)dy+\scrO_\epsilon(\|\hat{a}\|_{l^1}N^{1/2}E^{-\gamma_d+\epsilon}L^{-2\gamma_d+\epsilon})
\end{equation}
and if we introduce the density of the impurities $\rho=NL^{-d}$ we obtain as a condition for the existence of uniformly distributed eigenfunctions in terms of the three parameters the condition\footnote{The notation $\ll$ denotes ``much smaller than''.}
\begin{equation}
L \ll E^{\alpha_d}\rho^{-\beta_d} 
\end{equation}
where $$\alpha_d=\frac{2\gamma_d-\epsilon}{3d+4\gamma_d-\epsilon}, \quad \beta_d=\frac{1}{3d+4\gamma_d-\epsilon}.$$
\end{thm}

\begin{remark}{\bf Lower bound on the localization length.}\\
Our result implies that if we are still in the localized regime, the localization length must be larger than the size of the torus $L$. In particular we obtain the polynomial lower bound $$L_{loc}\gtrsim E^{\alpha_d}\rho^{-\beta_d}.$$

We remark that the dependence on $N$ in the error term prevents us from studying the limit $L\to\infty$ for a positive density of potentials in order to study the important problem of delocalization for random Schr\"odinger operators. However, for a different stochastic process, so-called ``random displacement models'', we are able to overcome this obstacle, which is the subject of the forthcoming paper \cite{U2}.
\end{remark}

\begin{remark}{\bf Strong coupling renormalization.}\\
We remark here that the theorem above in fact holds for a general superposition of Green's functions, where the spectral parameter $\lambda$ lies in the interval $(n_k,n_{k+1})$ for $n_k\in S'$. {\bf The exact position of $\lambda$ inside the interval is not important.} In particular our results apply to the strong coupling regime, sometimes studied in the physics literature, which requires a renormalization of the parameters of the self-adjoint extension (cf. for instance \cite{Shigehara} and \cite{U}, section 3, p. 5).
\end{remark}


\subsection*{Acknowledgements}
This work was largely completed as a Postdoc at the Institute of Theoretical Physics at CEA Saclay, and as a Postoc at the Laboratoire Paul Painlev\'e at the Universit\'e Lille 1, where I was supported in part by the Labex CEMPI (ANR-11-LABX-0007-01). I am particularly grateful to St\'ephane Nonnenmacher for numerous discussions and helpful suggestions that led to the improvement of this paper. I would also like to thank the Max Planck Institute of Mathematics for its hospitality.

\section{Background}

\subsection{Self-adjoint extension theory}\label{extensions}
Let $x_1,\cdots,x_N$ be distinct points on $\T^d=\R^d/\Z^d$, $d=2,3$. Denote $\x=\{x_1,\cdots,x_N\}$. This section will be concerned with the rigorous mathematical realization of the formal operator 
\begin{equation}\label{formal}
-\Delta+\sum_{j=1}^N \alpha_j \delta(x-x_j), \quad \alpha_1,\cdots,\alpha_N\in\R.
\end{equation}

Define $D_{\x}:=C^\infty_c(\T^d-\x)$ and consider the restricted Laplacian $H=-\Delta|_{D_\x}$. Denote the Green's function of the Laplacian on $\T^d$ by $$G_\lambda(x,y)=(\Delta+\lambda)^{-1}\delta(x-y).$$

The operator $H$ has deficiency indices $(N,N)$ and the deficiency spaces are spanned by the bases of deficiency elements $\{G_{\pm\i}(x,x_1),\cdots,G_{\pm\i}(x,x_N)\}$ respectively. There exists a family of self-adjoint extensions of $H$ which is parameterized by the group $U(N)$. We denote the self-adjoint extension of $H$ associated with a matrix $U\in U(N)$ by $-\Delta_{\x,U}$. 



\subsubsection{Spectrum and eigenfunctions}
As explained above there are three types of eigenfunctions of the operator $-\Delta_{\x,U}$. Old eigenfunctions, as well as generic (type 2b) and non-generic (type 2a) eigenfunctions.




Our results hold for both types of new eigenfunctions. Since type 2a eigenfunctions only occur with probability $0$ and do not feel the presence of all impurities, we will ignore them for the rest of the paper, and focus on the generic new eigenfunctions of type 2b.

To find the new eigenfunctions of the operator $-\Delta_{\x,U}$ we want to solve
\begin{equation}
(\Delta_{\x,U}+\lambda)g_\lambda=0.
\end{equation}

We may write $g_\lambda$ in the decomposition
\begin{equation}\label{decomp}
g_\lambda=f_\lambda+\left\langle v,\G_\i\right\rangle + \left\langle Uv,\G_{-\i}\right\rangle
\end{equation} 
where $\G_{\lambda}(x)=(G_\lambda(x,x_1),\cdots,G_\lambda(x,x_N)),v\in\C^N$ and $g_\lambda\in C^\infty_c(\T^d-\x)$.

So we have
\begin{equation}
(\Delta+\lambda)f_\lambda+(-\i+\lambda)\left\langle v,\G_\i\right\rangle + (\i+\lambda)\left\langle Uv,\G_{-\i}\right\rangle=0.
\end{equation}
We apply the resolvent $(\Delta+\lambda)^{-1}$, for $\lambda\not\in\sigma(-\Delta)$, and obtain
\begin{equation}
f_\lambda+\frac{-\i+\lambda}{\Delta+\lambda}\left\langle v,\G_\i\right\rangle + \frac{\i+\lambda}{\Delta+\lambda}\left\langle Uv,\G_{-\i}\right\rangle= 0
\end{equation}
By the repeated resolvent identity $$\frac{\mp\i+\lambda}{(\Delta+\lambda)(\Delta\pm\i)}=-\frac{1}{\Delta+\lambda}+\frac{1}{\Delta\pm\i}$$ we can rewrite this equation as
\begin{equation}\label{FNeq}
f_\lambda-\left\langle v,\G_\lambda-\G_\i\right\rangle - \left\langle Uv,\G_\lambda-\G_{-\i}\right\rangle = 0
\end{equation}
Furthermore, note that we can write more compactly $$\left\langle v,\G_\lambda-\G_\i\right\rangle + \left\langle v,U^{-1}(\G_\lambda-\G_{-\i})\right\rangle
=\left\langle v, \A_\lambda \right\rangle$$
where $\A_\lambda(x)=(\G_\lambda-\G_\i)(x) + U^{-1}(\G_\lambda-\G_{-\i})(x)$. 

Now, since $f_\lambda=\left\langle v,\A_\lambda \right\rangle\in C^\infty_c(\T^d-\x)$, we obtain the equations (set $x=x_k$ for $k=1,\cdots,N$)
\begin{equation}
\left\langle v, \A_\lambda(x_k) \right\rangle = 0, 
\quad k=1,\cdots, N,
\end{equation}
which we can rewrite as the matrix equation
\begin{equation}\label{speceqn}
M_\lambda\;v = 0
\end{equation}
where $F_\x(\lambda)=M_\lambda=(\A_\lambda(x_1),\cdots,\A_\lambda(x_N))$. 

So in order to find nontrivial solutions we need to solve the spectral equation 
\begin{equation}\label{spectral}
\det M_\lambda=0.
\end{equation}
We note that the determinant $F_\x(\lambda)$ is a meromorphic function of $\lambda$ with poles at the Laplacian eigenvalues, which we recall are given by the set $S=\{n \mid n=4\pi^2(x^2+y^2), \;x,y\in\Z\}=\{0=n_0<n_1<n_2<\cdots\}$. 

Each interval $(n_k,n_{k+1})$ contains up to $N$ (a. s. simple) random eigenvalues which are solutions of the spectral equation \eqref{spectral}. 
Given a solution $\lambda\in\sigma(-\Delta_{\x,U})$ the corresponding eigenfunction is given by
\begin{equation}
G^N_{\lambda,\x}(x)=\left\langle(\Id+U)v,\G_\lambda(x)\right\rangle=\sum_{j=1}^N d_{\lambda,j}(\x)G_\lambda(x,x_j)
\end{equation}
which can be seen by substituting identity \eqref{FNeq} in \eqref{decomp}. 

\subsection{Scaling}\label{scaling}

It can easily be seen that the formal definition of the operator $-\Delta_{U,\x}$ via the theory of self-adjoint extensions corresponds to the standard Laplacian $-\Delta$ acting on functions $f\in C^\infty(\T^d-\x)$ where $\Delta f+c_1\delta_{x_1}+\cdots+c_N\delta_{x_N}\in L^2(\T^d)$, where $c_j\in\C$, $j=1,\cdots,N$, and $f$ satisfies certain logarithmic boundary conditions at each of the points $x_j$ which only depend on the choice of the matrix $U$.

Let $f\in L^2(\T^d_L)$ and define by $g(y)=f(Ly)$ a function $g\in L^2(\T^d)$. Let $\x_L=(Lx_1,\cdots,Lx_N)$.
It can easily be seen that the eigenvalue problem $$(\Delta_{U,\x_L}+E)f=0$$ on the large torus $\T^d_L$ corresponds to the eigenvalue problem $$(L^{-2}\Delta_{U,\x}+E)g=0$$ on the standard torus $\T^d$. If, in the first problem we study eigenfunctions with bounded eigenvalue $E\leq E_0$ and the limit of large tori $L\nearrow\infty$, then in the second problem this corresponds to studying the high energy limit $\lambda=EL^2\to\infty$.

\section{Proof of Theorem \ref{deloc}}
We give the detailed proof here only in the critical case of two dimensions. The proof works exactly the same in three dimensions, however, instead of Lemma \ref{construct sequence} we have to use the subsequence constructed in \cite{Y} as well as the different exponent $\gamma_3$.

Let $\T^2=\R^2/\Z^2$. Let $\x=(x_1,\cdots,x_N)$. 
\begin{equation}
G^{N}_{\lambda,\x}(x)=\sum_{j=1}^N d_{\lambda,j}(\x) G_\lambda(x,x_j)
\end{equation}
where we fix the normalization
$$\sum_{j=1}^N|d_{\lambda,j}(\x)|^2=1.$$

Note that we will be interested in the spatial distribution of the $L^2$-normalized eigenfunctions $g^N_{\lambda,\x}:=G^N_{\lambda,\x}/\|G^N_{\lambda,\x}\|_2$ and its dependence on the random variable $\x$, which is independent of the choice of normalization of the superposition vector.

Let $e_\xi(x):=e^{2\pi\i\left\langle \xi,x \right\rangle}$, for $\xi\in\Z^2$. We have
$$G_\lambda(x,x_j)=\sum_{\xi\in\Z^2}c_\lambda(\xi)e_\xi(x-x_j)$$ where $c_\lambda(\xi)=(4\pi^2|\xi|^2-\lambda)^{-1}$, so
\begin{equation}
G^{N}_{\lambda,\x}(x)=\sum_{\xi\in\Z^2}D_{\lambda,\underline{x}}(\xi)e_{\xi}(x)
\end{equation}
where $$ D_{\lambda,\underline{x}}(\xi)=c_\lambda(\xi)d_\lambda(\x,\xi)$$ 
and $$d_\lambda(\x,\xi):=\sum_{j=1}^N d_{\lambda,j}(\x) e_\xi(-x_j).$$


\subsection{Approximation of Green's functions on thin annuli}

Let $a\in C^\infty(\T^2)$, $\lambda,L_0>0$ and $\zeta\in\
\Z^2$. Define the annulus $$A_\zeta(\lambda,L_0)=\{\xi\in\Z^2 \mid |4\pi^2|\xi-\zeta|^2-\lambda|\leq L_0\}$$ and we will use the notation $A(\lambda,L_0)=A_0(\lambda,L_0)$.

We introduce the truncated Green's function $$G^N_{\lambda,L_0,\x}(x,y)=\sum_{\xi\in A(n_k,L_0)}D_{\lambda,\underline{x}}(\xi)e_{\xi}(x).$$

Recall the circle law, $$\#\{\xi\in\Z^2 \mid |\xi|^2\leq X\}=\pi X+O_\epsilon(X^{\theta+\epsilon})$$ where the best known exponent $\theta=\frac{133}{416}$ is due to Huxley \cite{Huxley}.

In section \ref{approximation} we will prove the following Proposition.
\begin{prop}\label{split prop}
Denote $g^N_{\lambda,\x}=G^N_{\lambda,\x}/\|G^N_{\lambda,\x}\|_2$ and $g^N_{\lambda,L_0,\x}=G^N_{\lambda,L_0,\x}/\|G^N_{\lambda,\x}\|_2$
There exists a certain $\delta\in(\theta/2,1/2-\theta)$ and a subsequence of Laplacian eigenvalues $S'\subset S$, of density one, such that we have for the uniform random process $\x\in\T^{2N}$ and for each $\lambda\in(n_k,n_{k+1}), n_k\in S'$ sufficiently large, that there exists an event $\Omega_1\subset\T^{2N}$ with $\Prob(\Omega_1)\geq \frac{1}{7N}$ such that
\begin{equation}\label{approx}
\left\langle a g^N_{\lambda,\x}, g^N_{\lambda,\x} \right\rangle
=\left\langle a g^N_{\lambda,L_0,\x}, g^{N}_{\lambda,L_0,\x} \right\rangle+\scrO_\epsilon(\|a\|_\infty N^{1/2}\lambda^{-\delta+\theta/2+\epsilon})
\end{equation}
and $L_0=n_k^\delta$.
\end{prop}

\subsection{Uniformly distributed eigenfunctions}

We have the following result which proves the existence of uniformly distributed eigenfunctions at high energy with positive probability.
\begin{thm}\label{uni}
Fix $U\in U(N)$. Let $a\in C^\infty(\T^2)$. There exists a density $1$ subsequence $S'\subset S$ such that we have for any $\lambda\in(n_k,n_{k+1}),n_k\in S'$ and $n_k$ sufficiently large, with probability $\gtrsim\frac{1}{N}$,
\begin{equation}
\int_{\T^2}a(y)\Big|g^N_{\lambda,\x}(y)\Big|^2 dy = \int_{\T^2}a(y)dy+\scrO_{\epsilon}(\|\hat{a}\|_{l^1} N^{1/2}\lambda^{-\frac{17}{832}+\epsilon}).
\end{equation}
\end{thm}

The construction of the sequence $S'$ is almost identical to the one given in \cite{RU}, sections 5, 6 and 7. We recall this construction as a separate Lemma.
\begin{lem}\label{construct sequence}
Denote by $S=\{n \mid n=4\pi^2(x^2+y^2), x,y\in\Z\}=\{0=n_0<n_1<n_2<\cdots\}$ the set of Laplacian eigenvalues on the torus $\T^2=\R^2/\Z^2$, where we ignore multiplicities. 

There exists small $\epsilon>0$ and a subsequence $S'\subset S$ of density $1$ such that for all $n_k\in S'$:
\begin{itemize}
\item[(i)] $n_{k+1}-n_{k-1} \lesssim_{\epsilon'} n_k^{\epsilon'}$
\item[(ii)] $\forall\lambda\in(n_k,n_{k+1})\; \forall\zeta\neq0,|\zeta|\leq \lambda^\epsilon; \forall\xi\in A(n_k,L_0): |c_\lambda(\xi+\zeta)|\lesssim\frac{1}{L_0}$
\end{itemize}
\end{lem}
\begin{proof}
(i): To see this note that the elements of $S$, integers representable as sums of two squares, have mean spacing of order $\sqrt{\log n_k}$. Therefore the subsequences of $n_k$ s. t. $n_{k+1}-n_k\lesssim_\epsilon n_k^\epsilon$ and those $n_k$ s. t. $n_k-n_{k-1}\lesssim_{\epsilon'} n_k^{\epsilon'}$ are of density $1$ respectively. Consequently, their intersection is a subsequence of density $1$.

(ii): The proof is very similar to the construction in sections 6 and 7 of \cite{RU}. Hence we only summarize the idea of the proof here and for the details refer the reader to the appropriate sections in \cite{RU}. Note the additional factor $4\pi^2$ which is due to the fact that we consider the standard torus $\R^2/\Z^2$ rather than the scaled torus $\R^2/2\pi\Z^2$ considered in \cite{RU}.

In \cite{RU} we introduced for $\zeta\in\Z^2\setminus\{0\}$ the set of ``bad'' lattice points
$$\scrB_{\zeta}=\{\xi\in\Z^2 \mid |\left\langle\xi,\zeta\right\rangle|\leq|\xi|^{2\delta}\}$$
and we showed that the set of norms 
$$\scrN_\zeta=\{n\in S \mid \exists\xi\in\scrB_\zeta:\, n=|\xi|^2\}$$
is of density $0$ in $S$, more precisely (cf. \cite{RU}, p. 773, eq. (6.4)) 
$$\#\{n\in\scrN_\zeta, n\leq X\}\lesssim\frac{X^{1/2+\theta+\delta}}{|\zeta|}$$
and recall $\delta<1/2-\theta$.



In \cite{RU} we showed that for $\delta\in(\theta/2,1-\theta/2)$ and $\zeta\in\Z^2\setminus\{0\}$ the sequence $$S_\zeta=\{n\in S \mid  A_0(n,n^\delta)\cap \scrB_\zeta=\emptyset\}$$ is of full density in $S$. If $n\in S_{\zeta}$, then $\xi\in A(n,n^\delta)$ implies $\xi+\zeta\notin A(n,n^\delta)$ which is equivalent to $||\xi+\zeta|^2-n|\gtrsim n^\delta$ which implies $|c_\lambda(\xi+\zeta)|\lesssim \lambda^{-\delta}$ for $\lambda\in(n_k,n_{k+1})$ and $n_k\in S_\zeta$.

Now let $$S'=\{n\in S \mid \forall\zeta\in\Z^2\setminus\{0\},|\zeta|\leq n^\epsilon:\; A(n,n^\delta)\cap \scrB_\zeta=\emptyset\}$$ and $$\epsilon=\frac{1/2-\theta-\delta}{2}$$
and we have
\begin{equation}
\begin{split}
&\#\{m\in S'^c, m\leq X\}\\
=&\#\{m\in S, m\leq X \mid \exists\zeta\in\Z^2\setminus\{0\},|\zeta|\leq m^\epsilon:\; A(m,m^\delta)\cap \scrB_\zeta\neq\emptyset\}\\
\leq&\sum_{\substack{\zeta\in\Z^2\setminus\{0\}\\|\zeta|\leq X^\epsilon}}\#\{m\in S_\zeta^c \mid |m|\leq X\}\\
\lesssim& X^{1/2+\delta+\theta}\sum_{\substack{\zeta\in\Z^2\setminus\{0\}\\|\zeta|\leq X^\epsilon}}\frac{1}{|\zeta|}\\
=& X^{1/2+\delta+\theta}\sum_{n\in S,n\leq X^\epsilon}\frac{r_2(n)}{n^{1/2}}\lesssim X^{1/2+\delta+\theta+\epsilon}=o(X)
\end{split}
\end{equation}
which proves that the subsequence $S'$ is of density $1$ in $S$.

Now to conclude the proof, recall $L_0=n_k^\delta$, we assume that $\xi\in A(n_k,L_0)$, $\lambda\in(n_k,n_{k+1})$ and $n_k\in S'$. We thus have for any $\zeta\in\Z^2\setminus\{0\}$, $|\zeta|\leq n_k^\epsilon$ the lower bound
\begin{equation}
\begin{split}
||\xi+\zeta|-\lambda|
\geq& 2|\left\langle \xi,\zeta \right\rangle|-|||\xi|^2-\lambda|+|\zeta|^2|\\
>&2|\xi|^{2\delta}-\lambda^\delta+O(\lambda^{2\epsilon})\\
\asymp& \lambda^\delta
\end{split}
\end{equation}
which implies the bound $|c_\lambda(\xi+\zeta)|\lesssim\lambda^{-\delta}=1/L_0$.

\end{proof}

{\em Proof of Theorem \ref{uni}:}
We expand the test function $a\in C^\infty(\T^2)$ into a Fourier series $$a=\sum_{\zeta\in\Z^2}\hat{a}(\zeta)e_\zeta.$$

Recall that $\lambda\in(n_k,n_{k+1})$ for some $n_k\in S'$. Let $$a_\epsilon=\sum_{\substack{\zeta\in\Z^2\setminus\{0\}\\|\zeta|\leq n_k^\epsilon}}\hat{a}(\zeta)e_\zeta$$ and observe
$$\left\langle a g_{\lambda,\x}^N,g_{\lambda,\x}^N\right\rangle=\left\langle a_\epsilon g_{\lambda,\x}^N,g_{\lambda,\x}^N \right\rangle+O(\lambda^{-\infty})$$ since the Fourier coefficients of $a$ decay faster than any inverse power of $|\zeta|$, and thus $$\|a-a_\epsilon\|_\infty
\leq \sum_{\substack{\zeta\in\Z^2\setminus\{0\}\\|\zeta|> n_k^\epsilon}}|\hat{a}(\zeta)|\lesssim_k \lambda^{-k}, \,\forall k>0.$$

We compute, for $\zeta\neq0$,
\begin{equation}
\left\langle e_\zeta G^{N}_{\lambda,L_0,\x}, G^{N}_{\lambda,L_0,\x} \right\rangle
=\sum_{\xi\in A(n_k,L_0)}D_{\lambda,\underline{x}}(\xi)\overline{D_{\lambda,\underline{x}}(\xi+\zeta)}.
\end{equation}

Hence, we have, by Cauchy-Schwarz,
\begin{equation}
\frac{\left|\left\langle e_\zeta G^N_{\lambda,L_0,\x}, G^N_{\lambda,L_0,\x} \right\rangle\right|}{\|G^N_{\lambda,L_0,\x}\|_2^2}^2
\leq\frac{\sum_{\xi\in A(n_k,L_0)}|D_{\lambda,\underline{x}}(\xi+\zeta)|^2}{\sum_{\xi\in A(n_k,L_0)}|D_{\lambda,\underline{x}}(\xi)|^2}
\end{equation}

Before we continue with the estimation, let us define the following functions of the random variable $\x$
\begin{equation}
\begin{split}
\scrA_\zeta(\x)=
&\sum_{\substack{\xi\in A(n_k,L_0)\\ |\xi+\zeta|^2<n_k}}c_{n_k}(\xi+\zeta)^2|d_\lambda(\x,\xi)|^2\\
&+\sum_{\substack{\xi\in A(n_k,L_0)\\ |\xi+\zeta|^2>n_{k+1}}}c_{n_{k+1}}(\xi+\zeta)^2|d_\lambda(\x,\xi)|^2
\end{split}
\end{equation}
and
\begin{equation}
\scrB(\x)=\frac{|d_\lambda(\x,\xi_0)|^2}{(n_{k+1}-n_{k-1})^2}
\end{equation}
where $\xi_0\in\Z^2$ is such that $4\pi^2|\xi_0|^2=n_k$. 

Now, note that for $\lambda\in(n_k,n_{k+1})$, $n_k\in S'$
$$\sum_{\xi\in A(n_k,L_0)}|D_{\lambda,\x}(\xi+\zeta)|^2 \leq \scrA_\zeta(\x)$$
because $4\pi^2|\xi+\zeta|^2<n_k$ implies $\lambda-4\pi^2|\xi+\zeta|^2>n_k-4\pi^2|\xi+\zeta|^2$, and this in turn implies 
$$c_\lambda(\xi+\zeta)^2=(\lambda-4\pi^2|\xi+\zeta|^2)^{-2}<(n_k-4\pi^2|\xi+\zeta|^2)^{-2}=c_{n_k}(\xi+\zeta)^2$$ and the analogous argument applies to the sum over $4\pi^2|\xi+\zeta|^2>n_{k+1}$. Also note that since $n_k\in S'$ we have that $\xi\in A(n_k,L_0)$ implies, for $l=k,k+1$, $||\xi+\zeta|^2-n_l(\lambda)|\gtrsim\lambda^\delta$ (cf. bound (ii) of Lemma \ref{construct sequence}) and therefore $4\pi^2|\xi+\zeta|^2\neq n_{k+1},n_k$.

We also have the lower bound
$$\sum_{\xi\in A(n_k,L_0)}|D_{\lambda,\x}(\xi)|^2 \geq \scrB(\x)$$
because $\lambda-n_{k-1}\leq n_{k+1}-n_{k-1}$.

Again, by Cauchy-Schwarz,
\begin{equation}
\begin{split}
\Big|&\hat{a}(0)-\left\langle a_\epsilon g^N_{\lambda,L_0,\x},g^N_{\lambda,L_0,\x}\right\rangle\Big|^2\\
\leq& \Big|\sum_{\substack{\zeta\in\Z^2\setminus\{0\}\\|\zeta|\leq n_k^\epsilon}}\hat{a}(\zeta)\left\langle e_\zeta g^N_{\lambda,L_0,\x},g^N_{\lambda,L_0,\x} \right\rangle\Big|^2\\
\leq& \sum_{\substack{\zeta\in\Z^2\setminus\{0\} \\ |\zeta|\leq n_k^\epsilon}}|\hat{a}(\zeta)|\; \times
\sum_{\substack{\zeta\in\Z^2\setminus\{0\}\\ |\zeta|\leq n_k^\epsilon}}|\hat{a}(\zeta)|\Big|\left\langle e_\zeta g^N_{\lambda,L_0,\x},g^N_{\lambda,L_0,\x} \right\rangle\Big|^2\\
\leq & \|\hat{a}\|_{l^1}\frac{1}{\scrB(\x)}\underbrace{\sum_{\substack{\zeta\in\Z^2\setminus\{0\}\\|\zeta|\leq n_k^\epsilon}} |\hat{a}(\zeta)|\scrA_\zeta(\x)}_{:=\scrA_a(\x)}.
\end{split}
\end{equation}

We have the following proposition which we prove in section 4.
\begin{prop}\label{probabilities}
Let $\x=(x_1,\cdots,x_{N_0})$ be points from a stochastic process on $\T^2$, and denote its joint probability distribution by $\PP_{N_0}$. Let $\scrA:\T^{2N_0}\to\R_+$, $\scrB$ be as above and $C_0>0$.
We have that 
\begin{equation}
\scrA(\x)\leq C_0\E(\scrA), \text{with probability $>1-\frac{1}{C_0}$}
\end{equation}
and 
\begin{equation}\label{lower bound}
\scrB(\x)> \frac{1}{3}\E(\scrB), \text{with probability $>\frac{9}{14N_0}$.}
\end{equation}
\end{prop}

We also require the following proposition which is an immediate consequence of the identities \eqref{exp A} and \eqref{exp B} proven in section 4.
\begin{prop}
We have the following asymptotics, as $\lambda\in(n_k,n_{k+1})$ and $n_k\to\infty$,
\begin{equation}
\frac{\E(\scrA_\zeta)}{\E(\scrB)} \sim (n_{k-1}-n_{k+1})^2\Sigma(\zeta,n_k)
\end{equation}
where
$$
\Sigma(\zeta,n_k)=\sum_{\substack{\xi\in A(n_k,L_0)\\ |\xi+\zeta|^2<n_k}}c_{n_k}(\xi+\zeta)^2+\sum_{\substack{\xi\in A(n_k,L_0)\\ |\xi+\zeta|^2>n_{k+1}}}c_{n_{k+1}}(\xi+\zeta)^2.
$$
\end{prop}

The event $\Omega_1$ of Proposition \ref{split prop} implies the lower bound \eqref{lower bound}, since a lower bound on the $L^2$-norm is required in the proof (cf. section 4). Proposition \ref{probabilities} implies that there exists an event $\Omega_2$ with $\Prob(\Omega_2)\geq 1-\frac{1}{14N}$ such that $$\scrA_a(\x)\leq 14N\E(\scrA_a).$$

We hence have for $\x\in\Omega_2$, for any $n_k\in S'$ sufficiently large and $\lambda\in(n_k,n_{k+1})$, 
\begin{equation}\label{A/B bound}
\begin{split}
\frac{\scrA_a(\underline{x})}{\scrB(\underline{x})}
&\lesssim\frac{\E(\scrA_a)}{\E(\scrB)}\\
&\lesssim N \|\hat{a}\|_{l^1} (n_{k-1}-n_{k+1})^2\sum_{\substack{\zeta\in\Z^2\setminus\{0\}\\|\zeta|\leq n_k^\epsilon}}|\hat{a}(\zeta)|\Sigma(\zeta,n_k)\\
&\lesssim_{\epsilon'} N \|\hat{a}\|_{l^1}^2\lambda^{-2\beta+\epsilon'}
\end{split}
\end{equation} 
for some absolute constant $\beta=\delta-\theta/2$, which follows from the bound $$\forall\zeta\in\Z^2\setminus\{0\},|\zeta|\leq n_k^\epsilon:\quad\Sigma(\zeta,n_k)\lesssim\#A(n_k,L_0)/L_0^2\lesssim n_k^{\theta-2\delta}\sim \lambda^{\theta-2\delta}$$ and we recall $L_0=n_k^\delta\sim\lambda^\delta$ and the bound $\#A(n_k,L_0)\lesssim n_k^\theta\sim\lambda^\theta$ which is a consequence of the circle law. We also used the bound $n_{k+1}-n_{k-1}\lesssim_{\epsilon'} n_k^{\epsilon'}$.


Finally note that $\Prob(\Omega_1\cap\Omega_2)\geq\frac{1}{7N}-\frac{1}{14N}\geq\frac{1}{14N}$, and hence for $\x\in\Omega_1\cap\Omega_2$ both the approximation \eqref{approx} and the bound \eqref{A/B bound} hold. 

Therefore, we have with probability $\geq\frac{1}{14N}$
\begin{equation}
\begin{split}
\left\langle a g^N_{\lambda,\x},g^N_{\lambda,\x}\right\rangle=&\left\langle a_\epsilon g^N_{\lambda,\x},g^N_{\lambda,\x}\right\rangle+\scrO(\lambda^{-\infty})\\
=&\left\langle a_\epsilon g^N_{\lambda,L_0,\x},g^N_{\lambda,L_0,\x}\right\rangle+\scrO_{\epsilon'}(\|a\|_\infty N^{1/2}\lambda^{-\delta+\theta/2+\epsilon'})\\
=&\hat{a}(0)+\scrO_{\epsilon'}(\|\hat{a}\|_{l^1}N^{1/2}\lambda^{-\delta+\theta/2+\epsilon'})
\end{split}
\end{equation}
where we note that we get the best exponent with the choice $\delta=1/2-\theta-\epsilon$.

\hfill $\square$

\section{Proofs of the auxiliary results}

\subsection{Some expectation values}
We compute the expectation value of $|d_\lambda(\x,\xi)|^2$:
\begin{equation}
\begin{split}
&\E\left(|d_\lambda(\x,\xi)|^2\right)\\
=&\int_{\T^2}\cdots\int_{\T^2}|d_\lambda(\x,\xi)|^2 d\underline{x}\\
=& \int_{\T^2}\cdots\int_{\T^2}|\sum_{j=1}^N d_{\lambda,j}(\x) e_\xi(-x_j)|^2 d\underline{x}\\
=& \sum_{j,k=1}^N\int_{\T^2}\cdots\int_{\T^2} d_{\lambda,j}(\x) \overline{d_{\lambda,k}(\x)} e_\xi(x_k-x_j) dx_1 \cdots dx_N\\
=& \sum_{j=1}^N \E\left(|d_{\lambda,j}|^2\right)
+\sum_{\substack{1\leq j,k\leq N \\ j\neq k}}\widehat{(d_{\lambda,j}\overline{d_{\lambda,k}})}(\Xi_{j,k}).
\end{split}
\end{equation}
where $\Xi_{j,k}=(0,0,\cdots,\underbrace{\xi_1,\xi_2}_{\text{jth place}},\cdots,\underbrace{-\xi_1,-\xi_2}_{\text{kth place}},\cdots,0,0)$.

For convenience, denote $\FF_\xi(d_{\lambda,j}\overline{d_{\lambda,k}})=\widehat{(d_{\lambda,j}\overline{d_{\lambda,k}})}(\Xi_{j,k})$.

Recall the definition of the functions 
\begin{equation}
\scrA_\zeta(\x)=\sum_{\substack{\xi\in A(n_k,L_0)\\ |\xi+\zeta|^2<n_k}}c_{n_k}(\xi+\zeta)^2|d_\lambda(\x,\xi)|^2
+\sum_{\substack{\xi\in A(n_k,L_0)\\ |\xi+\zeta|^2>n_{k+1}}}c_{n_{k+1}}(\xi+\zeta)^2|d_\lambda(\x,\xi)|^2
\end{equation}
and
\begin{equation}
\scrB(\x)=c_{n_{k+1}}(\xi_0)^2|d_\lambda(\x,\xi_0)|^2, \quad |\xi_0|^2=n_{k-1},
\end{equation}
and in addition define
\begin{equation}
\scrC(\x)=\sum_{\substack{\xi\in A(n_k,L_0)^c\\ |\xi|^2<n_k}}c_{n_k}(\xi)^2|d_\lambda(\x,\xi)|^2+\sum_{\substack{\xi\in A(n_k,L_0)^c\\ |\xi|^2>n_{k+1}}}c_{n_{k+1}}(\xi)^2|d_\lambda(\x,\xi)|^2.
\end{equation}

Now, denote $\E_N+\FF_{N,\xi}=\sum_{j=1}^N\E\left(|d_{\lambda,j}|^2\right)+\sum_{j\neq k}\FF_\xi(d_{\lambda,j}\overline{d_{\lambda,k}})$ and note that our normalization implies $\E_N=1$.

We have
\begin{equation}
\begin{split}
&\E(\scrA_\zeta)\\
=&\sum_{\substack{\xi\in A(n_k,L_0)\\ |\xi|^2<n_k}}c_{n_k}(\xi+\zeta)^2\left(1+\FF_{N,\xi+\zeta}\right)
+\sum_{\substack{\xi\in A(n_k,L_0)\\ |\xi|^2>n_{k+1}}}c_{n_{k+1}}(\xi+\zeta)^2\left(1+\FF_{N,\xi+\zeta}\right)
\end{split}
\end{equation}
and
$$\E(\scrB)=c_{n_{k+1}}(\xi_0)^2 \left(1+\FF_{N,\xi}\right)$$
and
$$\E(\scrC)=\sum_{\substack{\xi\in A(n_k,L_0)^c\\ |\xi|^2<n_k}}c_{n_k}(\xi)^2\left(1+\FF_{N,\xi}\right)+\sum_{\substack{\xi\in A(n_k,L_0)^c\\ |\xi|^2>n_{k+1}}}c_{n_{k+1}}(\xi)^2\left(1+\FF_{N,\xi}\right).$$

We have the following Lemma.
\begin{lem}\label{offdiag}
The diagonal terms in the expectation value dominate the off-diagonal terms in the limit $|\xi|\to\infty$:
\begin{equation}
\sum_{j\neq k}\FF_\xi(d_{\lambda,j}\overline{d_{\lambda,k}})=o(1)
\end{equation}
\end{lem}
\begin{proof}

We have
\begin{equation}
\sum_{\substack{1\leq j,k\leq N \\ j\neq k}} \widehat{d_{\lambda,j}\overline{d_{\lambda,k}}}(\Xi_{j,k})
=o_{|\xi|\to\infty}(1)
\end{equation}
because $|\widehat{d_{\lambda,j}\overline{d_{\lambda,k}}}(\Xi_{j,k})|=o(\|d_{\lambda,j}\overline{d_{\lambda,k}}\|_2)$, as $|\xi|\to\infty$, in view of
$$\sum_{\xi\in\Z^2}|\widehat{d_{\lambda,j}\overline{d_{\lambda,k}}}(\Xi_{j,k})|^2
\leq \sum_{\xi\in\Z^{2N}}|\widehat{d_{\lambda,j}\overline{d_{\lambda,k}}}(\xi)|^2
=\|d_{\lambda,j}\overline{d_{\lambda,k}}\|_2^2\leq 1$$
where we recall our normalization of coefficients which ensures $\forall j: |d_{\lambda,j}|\leq 1$.
\end{proof}

As a consequence of the Lemma \ref{offdiag} we have the asymptotics (recall that $c_\lambda(\xi)=(4\pi^2|\xi|^2-\lambda)^{-1}$ is weighted near $\lambda$)
\begin{equation}\label{exp A}
\E(\scrA_\zeta)\\
\sim_{n_k\to\infty}\sum_{\substack{\xi\in A(n_k,L_0)\\ |\xi+\zeta|^2<n_k}}c_{n_k}(\xi+\zeta)^2
+\sum_{\substack{\xi\in A(n_k,L_0)\\ |\xi+\zeta|^2>n_{k+1}}}c_{n_{k+1}}(\xi+\zeta)^2
\end{equation}
and, where $\xi_0\in\Z^2$ is such that $4\pi^2|\xi_0|^2=n_{k-1}$,
\begin{equation}\label{exp B}
\E(\scrB)\sim_{n_k\to\infty}c_{n_{k+1}}(\xi_0)^2=\frac{1}{(n_{k+1}-n_{k-1})^2} 
\end{equation}
and
\begin{equation}\label{exp C}
\E(\scrC)\\
\sim_{n_k\to\infty}\sum_{\substack{\xi\in A(n_k,L_0)^c\\ |\xi|^2<n_k}}c_{n_k}(\xi)^2
+\sum_{\substack{\xi\in A(n_k,L_0)^c\\ |\xi|^2>n_{k+1}}}c_{n_{k+1}}(\xi)^2
\end{equation}



\subsection{Proof of Proposition \ref{probabilities}}
Let $N_0\in\N$ and $\scrA:\T^{2N_0}\to\R_+$. Also recall the specific function $\scrB:\T^{2N_0}\to\R_+$ that we defined above.

(i): Define the subset $A_1:=\{x\in\T^{2N_0} \mid \scrA(x)> C_0\E(\scrA)\}$. We denote the probability of the event $A_1$ occurring by $$\PP(A_1)=\int_{A_1}\PP_{N_0}(x)dx.$$
 
Now we have (recall $\scrA_a\geq0$ by definition): $$\E(\scrA)\geq \int_{A_1}\PP_{N_0}(x)\scrA(x)dx > C_0 \E(\scrA)\PP(A_1)$$ and therefore $\PP(A_1)<1/C_0$. This implies $$\Prob\{x\in\T^{2N_0} \mid \scrA_a(x)\leq C_0\E(\scrA)\}> 1-\frac{1}{C_0}.$$

(ii): Define the subset $A_2:=\{x\in\T^{2N_0} \mid \scrB(x)\leq\tfrac{1}{3}\E(\scrB)\}$. 

Fix $\delta_0\in(0,\frac{1}{28})$. There exists $\lambda_0>0$ such that for any $\lambda\geq\lambda_0$
\begin{equation}
\begin{split}
\scrB(\x)=c_{n_{k+1}}(\xi_0)^2|d_\lambda(\x,\xi_0)|^2\leq &c_{n_{k+1}}(\xi_0)^2(\sum_{j=1}^{N_0}|d_{j,\lambda}(\x)|)^2
\\
\leq &c_{n_{k+1}}(\xi_0)^2 N_0\sum_{j=1}^{N_0}|d_{j,\lambda}(\x)|^2\\
=&N_0c_{n_{k+1}}(\xi_0)^2\leq N_0(1+\delta_0)\E(\scrB)
\end{split}
\end{equation}
where the last bound follows from $\E(\scrB)=(1+o_{n_k\to\infty}(1))c_{n_{k+1}}(\xi_0)^2$.

We have
\begin{equation}
\begin{split}
(1-\PP(A_2))(1+\delta_0)N_0\E(\scrB)&\geq\int_{\T^{2N_0}\setminus A_2}\PP_{N_0}(x)\scrB(x)dx\\
&=\E(\scrB)-\int_{A_2} \PP_{N_0}(x)\scrB(x)dx\\
&\geq(1-\frac{1}{3}\PP(A_2))\E(\scrB)
\end{split}
\end{equation}
so we get 
$$\Prob\{\scrB(x)>\frac{1}{3}\E(\scrB)\}=1-\PP(A_2)\geq \frac{1-\tfrac{1}{3}}{N_0(1+\delta_0)-\tfrac{1}{3}}> \frac{9}{14N_0}$$
since $2/(3(1+\delta_0))>9/14$.

\subsection{Proof of Proposition \ref{split prop}}\label{approximation}

Let $a\in C^\infty(\T^2)$.
We will show that there exists a density $1$ subsequence $S'\subset S$ such that we have for $\lambda\in(n_k,n_{k+1})$, $L_0=\lambda^\delta$ and $\delta\in(\theta/2,\tfrac{1}{2}-\theta)$, where $\theta=\frac{133}{416}$ is the best known exponent in the circle law due to Huxley \cite{Huxley}, with probability $>\frac{9}{14N}-\frac{1}{2N}=\frac{1}{7N}$,
\begin{equation}\label{approx}
\frac{\left\langle a G^N_{\lambda,\x}, G^N_{\lambda,\x} \right\rangle}{\|G^N_{\lambda,\x}\|_2^2}
=\frac{\left\langle a G^N_{\lambda,L_0,\x}, G^{N}_{\lambda,L_0,\x} \right\rangle}{\|G^{N}_{\lambda,\x}\|_2^2}+\scrO_\epsilon(\|a\|_\infty N^{1/2}\lambda^{-\alpha+\epsilon})
\end{equation}
for $\alpha=\delta-\theta/2$.

We proceed with the proof of \eqref{approx}.
\begin{proof}
We split $g^N_{\lambda,\x}=g^N_{\lambda,L_0,\x}+g^N_{\lambda,R,\x}$, where $g^N_{\lambda,L_0,\x}=G^N_{\lambda,L_0,\x}/\|G^N_{\lambda,\x}\|_2$.

We then have the bound
\begin{equation}\label{split equation}
\begin{split}
\Big|&\left\langle a g^N_{\lambda,\x}, g^N_{\lambda,\x} \right\rangle -\left\langle a g^N_{\lambda,L_0,\x}, g^N_{\lambda,L_0,\x} \right\rangle\Big|\\
\leq& \Big|\left\langle a g^N_{\lambda,R,\x},g^N_{\lambda,R,\x}\right\rangle\Big|
+\Big|\left\langle a g^N_{\lambda,R,\x},g^N_{\lambda,L_0,\x}\right\rangle\Big|
+\Big|\left\langle a g^N_{\lambda,L_0,\x},g^N_{\lambda,R,\x}\right\rangle\Big|\\
\leq& \;\|a\|_\infty(\|g^N_{\lambda,R,\x}\|_2^2+2\|g^N_{\lambda,R,\x}\|_2)
\end{split}
\end{equation}
where $g^N_{\lambda,R}=(G^N_{\lambda,\x}-G^N_{\lambda,L_0,\x})/\|G^N_{\lambda,\x}\|_2$.

We estimate
\begin{equation}
\begin{split}
\|g^N_{\lambda,R,\x}\|_2^2=&\frac{\| G^N_{\lambda,\x}-G^N_{\lambda,L_0,\x}\|_2^2}{\| G^N_{\lambda,\x}\|_2^2}\\
<& \frac{\| G^N_{\lambda,\x}-G^N_{\lambda,L_0,\x}\|_2^2}{\| G^N_{\lambda,L_0,\x}\|_2^2} 
=\frac{\sum_{\xi\in A(n_k,L_0)^c}|D_{\lambda,\x}(\xi)|^2}{\sum_{\xi\in A(n_k,L_0)}|D_{\lambda,\x}(\xi)|^2}
\end{split}
\end{equation}

We want to bound $\|g^N_{\lambda,R,\x}\|_2$ in terms of an expression which does not depend on $\lambda=\lambda(\x)$. 

We thus have
\begin{equation}
\begin{split}
&\sum_{\xi\in A(n_k,L_0)^c}|D_{\lambda,\x}(\xi)|^2 \\
\leq &\sum_{\substack{\xi\in A(n_k,L_0)^c \\ |\xi|^2<n_k}}c_{n_k}(\xi)^2
|d_\lambda(\x,\xi)|^2
+\sum_{\substack{\xi\in A(n_k,L_0)^c \\ 
|\xi|^2>n_{k+1}}}c_{n_{k+1}}(\xi)^2|d_\lambda(\x,\xi)|^2\\
= &\;\scrC(\x)
\end{split}
\end{equation}
and, because $|n_{k-1}-\lambda|\leq |n_{k-1}-n_{k+1}|$,
\begin{equation}
\sum_{\xi\in A(n_k,L_0)}|D_{\lambda,\x}(\xi)|^2 
\geq \sum_{|\xi|^2=n_{k-1}}c_{n_{k+1}}(\xi)^2|d_\lambda(\x,\xi)|^2
\geq \scrB(\x).
\end{equation}

Using Proposition \ref{probabilities}, with probability $>\frac{9}{14N}-\frac{1}{2N}=\frac{1}{7N}$ we have $\scrC(\x)< 2N\E(\scrC)$ and $\scrB(\x)>\frac{1}{3}\E(\scrB)$. And thus, for $\lambda\in(n_k,n_{k+1})$, $n_k\in S'$, in view of the identities \eqref{exp B} and \eqref{exp C}, with probability $\gtrsim \frac{1}{N}$
\begin{equation}
\begin{split}
&\frac{\| G^N_{\lambda,\x}-G^N_{\lambda,L_0,\x}\|_2^2}{\| G^N_{\lambda,L_0,\x}\|_2^2}\\
&\leq\frac{\scrC(\x)}{\scrB(\x)}\lesssim\frac{\E(\scrC)}{\E(\scrB)}\\
&\lesssim N(n_{k+1}-n_{k-1})^2\left\{\sum_{\substack{\xi\in A(n_k,L_0)^c\\ |\xi|^2<n_k}}c_{n_k}(\xi)^2+\sum_{\substack{\xi\in A(n_k,L_0)^c\\ |\xi|^2>n_{k+1}}}c_{n_{k+1}}(\xi)^2\right\}\\
&\lesssim_\epsilon N\lambda^{-2\alpha+\epsilon}\\
&\lesssim N^{1/2}\lambda^{-\alpha+\epsilon} \quad \text{if $N\lambda^{-2\alpha+\epsilon}\ll 1$}
\end{split}
\end{equation}
for some absolute constant $\alpha=\delta-\theta/2$. To see this, note that the lattice point sum is bounded by the term $1/L_0+\lambda^\theta/L_0^2$, where $L_0=\lambda^\delta$ (cf. the proof of Lemma 5.1 in \cite{RU}), and we note that $\theta>1/4$ implies $\delta<1/2-\theta<\theta$ which is equivalent to $\delta>2\delta-\theta$ and thus $2\alpha=\min\{\delta,2\delta-\theta\}=2\delta-\theta$.
\end{proof}

\end{document}